\documentclass[11pt,oneside,a4paper]{article}
\usepackage{syntonly}
\usepackage{geometry}
\usepackage{mathrsfs}
\usepackage{makeidx}
\makeindex
\usepackage{amsmath}
\usepackage{amssymb}
\usepackage{amsmath}
\usepackage{amsfonts}
\usepackage{amssymb}

\usepackage{graphicx}
\usepackage{epstopdf}

\newtheorem{theorem}{Theorem}

\newtheorem{corollary}{Corollary}
\newtheorem{definition}{Definition}

\newenvironment{proof}[1][Proof]{\begin{trivlist}
\item[\hskip \labelsep {\bfseries #1}]}{\end{trivlist}}

\newcommand{\qed}{\nobreak \ifvmode \relax \else
      \ifdim\lastskip<1.5em \hskip-\lastskip
      \hskip1.5em plus0em minus0.5em \fi \nobreak
      \vrule height0.30em width0.4em depth0.25em\fi}
\newcommand*{\sign}{\mathop{\mathrm{sign}}\nolimits}

\author{\thanks{This project has received funding from the European Union's Horizon 2020 research and innovation programme under the Marie Sk\l{}odowska-Curie grant agreement No 731143}\,\ Willem L. Fouch\'e \thanks{corresponding author e-mail:fouchwl@gmail.com}  \,\ and  Safari Mukeru \thanks{e-mail: mukers@unisa.ac.za}\\
\footnotesize{\em Department of Decision Sciences, School of Economic  and Financial Sciences}\\ 
\footnotesize{University of South Africa, P. O. Box 392, Pretoria, 0003. South Africa}}

\title{On local times of Martin-L\"of random Brownian motion }
\date{}

\begin{document}

\maketitle

\pagenumbering{arabic}

\abstract{
In this paper we study the local times of Brownian motion  from the point of view of algorithmic randomness. We introduce the notion of effective local time and show that {\it any} path which is Martin-L\"of random with respect to the Wiener measure  has continuous effective local times at every computable point. 
Finally we obtain a new simple  representation of classical Brownian local times, computationally expressed.}

{\bf Keywords:} Brownian motion, Kolmogorov complexity, Martin-L\"of randomness, local times, effective  stochastic integrals\\ 
{\bf 2010 Mathematics Subject Classification: 03D32, 60H05, 60J65.} 

\section{Introduction}

A Brownian motion on the unit interval is a real-valued function $X: (t,\omega) \mapsto X(t,\omega)$ defined on $[0,1]\times \Omega$, where $\Omega$ is the underlying space of some probability space, such that $X (0,\omega)=0$, the function $t \mapsto X(t,\omega)$ is continuous for any $\omega$, and for any finite sequence $0< t_1 < \ldots < t_n \leq 1$, the functions $\omega \mapsto X(t_1,\omega),X(t_2,\omega)-X(t_1,\omega), \cdots, X(t_n,\omega)-X(t_{n-1},\omega)$ are random variables  statistically independent and normally distributed with means $0$ and variances $t_1,t_2-t_1,\cdots,t_n-t_{n-1}$, respectively.
In this paper, we shall consider the canonical Brownian motion $X$ defined by $X(t,\omega) = \omega(t)$ where $\Omega = C[0, 1]$ is the set of continuous real functions defined on $[0, 1]$ and vanishing at the origin, with its uniform norm topology and endowed with the {\it Wiener measure} $\mathbb{P}$. The random variable $\omega \mapsto X(t,\omega)$ will be denoted $X(t)$. The Wiener measure $\mathbb{P}$ is the unique Borel probability measure on $C[0, 1]$ such that for any $0 < t_1 < \ldots < t_n \leq 1$ and any Borel subset $A$ of $\mathbb{R}^n$, 
   \begin{eqnarray}\label{ewqres1}
\mathbb{P}\left[\omega\in C[0, 1]: (\omega(t_1), \ldots, \omega(t_n)) \in A \right]= \int_A \frac{e^{-\frac{1}{2} x^T Q^{-1} x}}{(2\pi)^{n/2} (\det Q)^{1/2}} dx
   \end{eqnarray}
   where $Q = (q_{ij})$ is the $n\times n$ matrix defined by $q_{ij} = t_i$, for $i\leq j$.  \\ If $Y$ is another Brownian motion on a probability space $(\Omega\,', \mathcal{G}, \mathbb{P}\,')$, then for any Borel subset $A$ of ~$C[0, 1]$
      \begin{eqnarray} \label{123wedswa12}
      \mathbb{P}\,'(\{\omega: Y(\omega) \in A\}) = \mathbb{P}(A)
      \end{eqnarray}
where $Y(\omega)$ is the function $[0, 1] \to \mathbb{R}$, $t \to Y(t,\omega)$.     
Discussions on the construction of Brownian motion can be found in \cite{ito_mckean, Morters_Peres}. 

The {\it occupation measure of the Brownian motion $X$ up to time $t$} is the random Borel measure defined by
$$\mu(t,\omega, A) = \lambda\{s\in [0, t]: X(s,\omega) \in A\},\, \, A \mbox{ Borel in }\mathbb{R},\,\,\omega\in \Omega.$$
Here $\lambda$ is the Lebesgue measure. 
 
 L\'evy \cite{Levy_48} proved that for almost all $\omega \in \Omega$, the occupation measure 
$\mu(t, \omega,.)$ is absolutely continuous, that is, there exists a function $$\mathscr{L}(t,\omega,.):\mathbb{R} \to \mathbb{R},\,\,x \mapsto \mathscr{L}(t, \omega, x)$$ such that
$$\mu(t, \omega, A) = \int_A \mathscr{L}(t, \omega, x) dx,\,\, (A \mbox{ Borel  in } \mathbb{R}).$$ The number $\mathscr{L}(t, \omega,x)$ is called the ``local time of $\omega$ at $x$ up to time $t$''. L\'evy called it the ``mesure du voisinage'' and it represents the ``time that the Brownian path $\omega$ spends at the point $x$ during the time interval $[0, t]$''. It is clear, by the Lebesgue density theorem that, for almost every $x\in \mathbb{R}$ (with respect to the Lebesgue measure),
\begin{eqnarray} \label{sdsd343ef}
\mathscr{L}(t, \omega, x) = \lim_{\epsilon \to 0^+} \frac{1}{2\epsilon} \lambda\{s \leq t: |X(s, \omega)-x|\leq \epsilon\}.\end{eqnarray}
Trotter \cite{Trotter}  proved later that  the occupation measure $\mu(t, \omega,.)$ has continuous density for almost all $\omega$, that is, $\mathscr{L}(., \omega,.): [0, 1] \times \mathbb{R}, (t,x) \to \mathscr{L}(t, \omega, x)$ is continuous. 
(See also \cite[Theorem 6.19]{Morters_Peres}.)
This has the implication that, for every $x \in \mathbb{R}$, almost surely, for all $t \in [0, 1]$,
   \begin{eqnarray} 
\mathscr{L}(t, \omega, x) = \lim_{\epsilon \to 0^+} \frac{1}{2\epsilon} \lambda\{s \leq t: |X(s, \omega)-x|\leq \epsilon\}.\end{eqnarray}

 We shall consider the local time at 0  and it will be denoted $\mathscr{L}(t, \omega)$ instead of $\mathscr{L}(t, \omega, 0)$. That is, 
\begin{eqnarray} \label{sdsd343wed3}
\mathscr{L}(t, \omega) = \lim_{\epsilon \to 0^+} \frac{1}{2\epsilon} \lambda\{s \leq t: |X(s, \omega)|\leq \epsilon\}.
\end{eqnarray}

The notion of local time has been extended to many stochastic processes. 
It  is an important key concept in the study of fine properties of stochastic processes, for example fractal dimensions of level sets, regularity of sample paths.  We refer to Geman and Horowitz \cite{Geman_Horowitz}, Xiao \cite{Xiao} and references therein for further background.  
The interplay between roughness of a function and the continuity of its local times is summarized in Table 1 in \cite{Geman_Horowitz}. For example, in the case of Gaussian processes, the smoothness of local times implies an  "irregularity "of its sample paths (Berman \cite{Berman_1972}).



We can, of course, associate with any real-valued function on the unit interval,  an occupation density, and ask wheter it is absolutely continuous with respect to Lebesgue measure. In this case the Radon-Nykodym derivative will be its local time. 
As indicated by Geman and Horowitz \cite{Geman_Horowitz},  it is an interesting open problem to investigate the existence of local times of particular classical functions such as the Weierstrass nowhere differentiable function and to determine which functions representable as Fourier series  are such that the occupation measure is absolutely continuous  and to compute the local times in terms of the Fourier coefficients.

One  aim of this project, of which this paper is the first of two, is to show that, relative to the halting problem, we can compute a continuous function on the unit interval, that does have local time, in the sense that the occupation density of this function is absolutely  continuous with respect to Lebesgue measure. It will turn out we can find such a function which is the complex oscillation  $\Phi(\Omega)$ associated with a Chaitin real $\Omega$  with $\Phi$ as constructed in Fouch\'e \cite{fo:2}.


Using the notion of Kolmogorov-Chaitin complexity, Asarin and Prokovskii \cite{ap:1} defined a subset of $C[0, 1]$ of full Wiener measure  independently of any specific property of Brownian motion. 
They considered the set of functions on $[0, 1]$ that can be uniformly approximated by sequences of piece-linear functions encoded by finite binary strings of high Kolmogorov-Chaitin complexity  and observed that such a class turns out to be, in later language,  exactly the class of Martin-L\"of random elements in the computable probability space $(C[0, 1], \mathbb{P})$ . Such  functions  are now referred to as {\it Martin-L\"of random Brownian paths} (ML Brownian paths) or {\it complex oscillations}.  
 
 It has been  observed in the past two decades  that {\it each}  Martin-L\"of random Brownian path or complex oscillation has many important properties of classical Brownian motion. For example, each complex oscillation satisfies the law of iterated logarithms \cite{fo:4}, the modulus of continuity of classical Brownian motion \cite{fo:1},  the doubling property of Hausdorff and Fourier dimension \cite{fomuda} and It\^o's lemma \cite{Mukeru_JOC_2014}.

Complex oscillations satisfy various other computable  properties \cite{fo:1, fo:5, KHNe:1, KHNe:2}.
 

It should be noted that in  \cite{fomuda}, the authors together with George Davie proved that  for {\it each} complex oscillation $\omega$, the occupation measure $\mu(t,\omega,.)$ of $\omega$ up to time $t$ is such that its Fourier transform 
$$\hat \mu(t,\omega, u) = \int_{0}^t \exp(i\, u\,\omega(s)) ds,\,\,u \in \mathbb{R} $$
satisfies 
$$|\hat \mu(t,\omega,u)|^2 = O(|u|^{-2+\epsilon}),\,\,u\to \infty $$ for all $\epsilon >0$. 
This has the implication that $\hat \mu(t,\omega,.) \in L^2(\mathbb{R})$  and by a standard argument, $\mu(t,\omega,.)$ is absolutely continuous  and
its Radon-Nikodym derivative $\mathscr{L}(t,\omega,.)\in L^2(\mathbb{R})$.  Hence for {\it almost every}  $x \in \mathbb{R}$ (with respect to the Lebesgue measure), 
    \begin{eqnarray} \label{233wea12q}
\mathscr{L}(t,\omega,x) = \lim_{\epsilon \to 0^+} \frac{\mu(x-\epsilon, x+\epsilon)}{2 \epsilon} = \lim_{\epsilon \to 0^+} \frac{1}{2\epsilon} \lambda\{s \leq t: |\omega(s)-x|\leq \epsilon\}.
\end{eqnarray} 
Thus, {\it each} complex oscillation admits local times at {\it almost} every level $x\in \mathbb{R}$.   
But nothing is known about the existence of local times $\mathscr{L}(t, \omega, x)$ at a {\it particular given point}  $x\in \mathbb{R}$. 
  
We shall prove that for {\it any} complex oscillation $\omega$ and {\it any} computable real number $x$,  and $\epsilon_n = 2^{-n},\,\, n \in \mathbb{N}$, 
the limit
 $$L(t, \omega, x) = \lim_{n \to  +\infty} \frac{1}{2 \epsilon_n} \lambda\{s \leq t: |\omega(s)-x|\leq \epsilon_n\}$$ exists and the function 
 $$L(.,\omega, x): [0, 1] \to [0, +\infty),\,\, t\mapsto L(t,\omega,x) $$  is continuous. 

It will turn out that  $L$ is indeed the local time of $\omega$ as defined in real analysis.

In the sequel  of this paper, which will be written in collaboration with George Davie, we  shall obtain effective versions of some classical results on local times of Brownian motion.   We shall prove that for any complex oscillation $\omega$, the effective local time $L$ of $\omega$ at the origin, denoted $L(t, \omega)$, is 
the maximum function of another complex oscillation $\tilde \omega$ obtained as a pointwise stochastic integral with respect to $\omega$.  This solves an open question in \cite{fo:2}. The fact that  the association $\omega \rightarrow \tilde\omega $ is in fact "layerwise computable " will be seriously discussed.  
 
 It is shown in \cite{fo:2} that each complex oscillation is uniquely determined by a Kolmogorov-Chaitin complex infinite binary string (also a $KC$-string). 
We prove that one can uniformly approximate the effective local times of a complex oscillation using an approximation of the complex oscillation itself. Given a complex oscillation $\omega$ defined by a $KC$-string $\alpha$, one can use a sequence of finite substrings of $\alpha$ to uniformly approximate both $\omega$ and its effective local time at the origin. Finally we show the effective local time of a complex oscillation $\omega$ can be approximated by the sum of absolute values of $\omega$ at dyadic rationals where $\omega$ changes its sign. As far as we know, this is a new result even in classical probability theory.

 \section{Preliminaries}


The metric space  $\Omega:=C_0[0, 1]$ is the space  of continuous functions $x:[0, 1] \to \mathbb{R}$ such that $x(0) = 0$) topologized by the uniform norm $$\|x\| = \sup\{|x(t)|: 0 \leq t \leq 1\}.$$  Let $\mathscr{S}$ be the set of functions $x \in \Omega$ that are piecewise  linear with rational coordinates for the points of non-differentiability (or corner points), that is, if $t$ is a point of non-differentiability of the function $x$, then both $t$ and $x(t)$ are rational numbers. Clearly,  $\mathscr{S}$ is a countable dense subset of $\Omega$. We shall consider a computable  enumeration 
    $$ x_1, x_2, \ldots, x_n \ldots $$ of all the elements of $\mathscr{S}$
 and assume it fixed throughout. The distance mapping 
      $$d: \mathscr{S}  \times \mathscr{S} \to \mathbb{Q},\,\, (x_i, x_j) \to d(x_i, x_j) = \|x_i - x_j\|,\,\,\,\,i, j = 1, 2 \ldots $$ is obviously computable in the sense that there exists an algorithm that on input $(i, j)$ yields the number $d(x_i, x_j)$. 
      
In the language of computability analysis, this means that the pair $(\mathscr{S}, d)$ defines a computable metric space of the metric space $\Omega$. (See for instance Hoyrup and Rojas \cite{horo:1}.)

For each $x \in \mathscr{S}$ and $r \in \mathbb{Q}$, let $B(x, r)$ be the ball in $\Omega$ with centre $x$ and radius $r$:
 $$B(x, r) = \{y \in \Omega: \|x -y\| < r\}.$$ Such balls are called {\it ideal} balls. 
The class of all ideal balls is obviously a countable set and it is dense in $\Omega$. We shall fix an effective listing  of all the ideal balls: $$B_1, B_2, \ldots, B_n, \ldots$$

Any open subset of $\Omega$ is a countable union of ideal balls. A subset $U \subset \Omega$ is an {\it effectively } open set if there is an algorithm that yields a sequence of integers $i_1, i_2, \dots, i_n, \dots$ such that 
   $$U = B_{i_1} \cup B_{i_2} \cup \ldots$$
An effectively open set is in fact a $\Sigma_1^0$ set. 

A sequence $U_1, U_2, \ldots$ of open subsets of $\Omega$ is called a {\it uniformly effective} sequence if there is an algorithm that, on input $i \in \mathbb{N}$ yields a sequence $i_1, i_2, \ldots, i_n, \ldots$ such that
    $$U_i = B_{i_1} \cup B_{i_2} \cup \ldots$$
(where $B_1, B_2, \ldots$ is the fixed sequence of ideal balls).
 \paragraph{Examples}
 \begin{enumerate}
   \item[(1)] In general, for fixed rationals $\alpha_1, \ldots, \alpha_n$, $ t_1, t_2, \ldots, t_n$, $a$, the subset 
       $$U = \{x\in \Omega: \alpha_1 x(t_1) + \ldots+ \alpha_n x(t_n) < a\}$$ is  effectively open as it can be seen from the following algorithm that expresses $U$ as a union of ideal balls. 
  Start with $L = \emptyset$ and $i = 1$ and consider the ideal ball $B_i$ (with centre  $x_i$ and radius $r_i$). If $\alpha_1 x_i(t_1) + \ldots+ \alpha_n x_i(t_n) < a$ and $r_i < |a - \alpha_1 x_i(t_1) + \ldots+ \alpha_n x_i(t_n)|$, then update $L\leftarrow L \cup \{i\}$. Repeat for $i \leftarrow i+1$. Then clearly $U = \cup_{i \in L} B_i.$ 
     \item[(2)]   Let $$U_n = \{x \in \Omega: \max_{1 \leq k \leq n} |x(k/n)| <1\},\,\,n = 1, 2, \ldots$$ Then $(U_n)$ is a  uniformly effective sequence of open sets. 
     Indeed, consider the following algorithm:
 Consider a list $A$ of all the elements of $\mathbb{N} \times \mathbb{N}$, for instance,
 $A = \{(1, 1), (1, 2), (2, 1), (1, 3), (2, 2), (3, 1), \ldots\}$.  
Starting with $i = 1$, consider the $i$th element of A, say $(j, n)$ and the ideal ball $B_j$ (assume that its centre is $ x_j $ and its radius is $r_j$). If $\max_{1 \leq k \leq n} |x_j(k/n)| <1$ and $r_j < |1 - \max_{1 \leq k \leq n} |x_j(k/n)||$ then include $B_i$ on the list of ideal balls for $U_n$. Update $i \leftarrow i+1$. 
This algorithm yields a list of ideal balls contained in $U_n$ for all $n =1, 2, \ldots$. 
 \end{enumerate}
 
 A Martin-L\"of test on $\Omega$ with respect to the Wiener measure $\mathbb{P}$ (or simply a Martin-L\"of test) is a sequence of uniformly effective open sets $(U_n)$ of $\Omega$ such that $\mathbb{P}(U_n) \leq 2^{-n}$ for all $n$. A path $x\in \Omega$ passes the test $(U_n)$ if $x\notin \cap_n U_n$. A path $x \in \Omega$ is called a {\it Martin-L\"of random Brownian path} (MLB path) if it passes all Martin-L\"of tests. 


The set of finite words or strings over the alphabet $\{0, 1\}$ is denoted by $\{0, 1\}^*$. Given a finite word or string $a$ over the alphabet  $\{0, 1\}$, that is,  $a \in \{0, 1\}^*$,  $|a|$ denotes the length of $a$. For $\alpha = \alpha_0 \alpha_1\ldots$ an infinite word over the alphabet $\{0, 1\}$, we denote by $\bar\alpha(n)$ the word $\alpha_0 \alpha_1\ldots \alpha_{n-1}.$ The Cantor space $\{0, 1\}^\mathbb{N}$ is denoted by $\mathscr{N}$. For any word $s \in \{0, 1\}^*$, $[s]$ is the interval $\{\alpha \in \mathscr{N}: \bar \alpha(n) = s\}$. The Lebesgue measure on $\mathscr{N}$ is denoted by $\lambda$ and $\lambda([s]) = 2^{-|s|}$ where $|s|$ denotes the length of the word $s$.  
\begin{definition}
A function $\nu:\{0, 1\}^* \to [0, 1]$ is called a recursively  enumerable semi-measure if
$$\sum_{a\in \{0, 1\}^*} \nu(a) \leq 1$$ and 
for  $q\geq 0$, with $q$ rational,  and $a \in \{0, 1\}^*$, the relation $\nu(a) > q$ is recursively enumerable in $a$ and $q$ in the sense that there exists an algorithm that, on input $a, q$ yields ``yes'' if indeed $\nu(a)>q$.
\end{definition}

\begin{theorem}
There exists a function $K:\{0, 1\}^* \to \mathbb{N}$ such that the function $\{0, 1\}^* \to [0, 1]: s\mapsto2^{-K(s)}$ is a recursively enumerable semi-measure  and for any other recursively enumerable semi-measure $\nu$, there is a constant $D$ such that $$K(a) \leq -\log_2 \nu(a) + D, $$
for all $a\in \{0, 1\}^*$. 
\end{theorem}
We shall fix one such function $K$ and for $a \in \{0, 1\}$, we shall  $K(a)$ is called the Kolmogorov complexity of the word $a$. For any other such function $K_1$,  $K_1(a) = K(a) + O(1)$ for all $a \in \{0, 1\}^*$ (see \cite{fo:1} and references therein for more discussions). 
An infinite binary string $\alpha$ is Kolmogorov-Chaitin complex if there exists $d$ such that for all $n$, $K(\bar \alpha(n)) \geq n-d.$ We call $d$ an incompressibility coefficient of $\alpha$. We will denote the set of Kolmogorov-Chaitin complex infinite binary strings by $KC$ and refer to its elements as $KC$-strings. It is well-known that $KC$ has Lebesgue measure 1.  
 We recall that $KC$-strings are exactly the Martin-L\"of random elements of the Cantor space. 
More background on Kolmogorov complexity can be found in \cite{ch:2, Li_Vitanyi, marlo:1, nie:1}.

    
For $n \geq 1$, we write $\mathscr{C}_n$ for the class of continuous functions on the unit interval that vanish at $0$ and are linear with slopes $\pm\ n^{-1/2}$ on the intervals $I_j = [(j-1)n^{-1},j n^{-1}]$, $j = 1,2, \ldots, n$. For every $\omega\in \mathscr{C}_n$, one can uniquely associate a binary string $a =a_1\cdots a_{n}$ by setting $a_j=1$ or $a_j =0$ according to whether $\omega$ increases or decreases on the interval $I_j$. The  string $a$ is called the code of $\omega$. Conversely any binary string $a\in \{0,1\}^*$  uniquely determines an element of $\mathscr{C}_n$, denoted $\psi(a)$.  
\begin{definition} \label{asdw34w3a}
A sequence $(\omega_n)$ in $C[0,1]$ is called a complex sequence  if $\omega_n \in \mathscr{C}_n$ for each $n$ and there is a constant $d > 0$ such that $K(\alpha_n) \geq n-d$ for all $n$, where $\alpha_n$ is the code  of $\omega_n$.  A function $\omega \in C[0,1]$ is a \emph{complex oscillation} if there is a complex sequence $(\omega_n)$ such that $\|\omega -\omega_n\|= \sup_{t\leq 1} |\omega(t) - \omega_n(t)|$ converges effectively to $0$ as $n \rightarrow \infty$ in the sense 
that there exists an algorithm that on input $m \in \mathbb{N}$ yields a natural number $m_0$ such that $\|\omega -\omega_n\|  \leq (m+1)^{-1}$ for all $m \in \mathbb{N}$ and $n > m_0$. 
\label{definition:ap}
\end{definition}  
We will denote the set of complex oscillations by $\mathscr{C}$. Asarin and Prokovskii \cite{ap:1} proved that  $\mathscr{C}$ is exactly the class of Martin-L\"of random elements of the computable probability space $(C[0, 1], \mathbb{P})$ where $\mathbb{P}$ is the Wiener measure. Consequently, $\mathscr{C}$ has full Wiener measure, that is, $\mathbb{P}(\mathscr{C}) =  1$. 

Fouch\'e \cite{fo:2} proved that any complex oscillation is completely determined by a single infinite binary $KC$-string and conversely each  such string determines a complex oscillation. 

In order to explore local properties of complex oscillations, Fouch\'e \cite{fo:1} introduced the  notion of effective generating sequence of subsets of $C[0, 1]$. We shall consider the following notations. For a subset $F$ of $C[0, 1]$ and $\epsilon>0$, $\bar F$ is the topological closure of $F$ and $O_\epsilon(F)$ is the $\epsilon$-neighbourhood of $F$,
 $$O_\epsilon(F) = \{f \in C[0, 1]: (\exists g \in F) (\|f- g\| < \epsilon)\}.$$
Conventionally we shall write $F^1 = F$ and $F^0$ the complement of $F$ in $C[0, 1]$. Throughout $\Sigma$ denotes the $\sigma$-algebra of $C[0, 1]$.
\begin{definition}\label{sew12dwe23}
A sequence $\mathscr{F}_0  = \{F_1, F_2, \ldots, F_n, \ldots\}$ in $\Sigma$ is an effective generating sequence if the following conditions are satisfied:
 \begin{itemize}
 \item[$(1)$] for each $F \in \mathscr{F}_0$ and  $\epsilon >0$, if  $G \in \{F, F^0, O_\epsilon(F), O_\epsilon(F^0)\}$, then $\mathbb{P}(\bar G) = \mathbb{P}(G).$
 \item[$(2)$] there is an algorithm that yields, for each finite sequence of integers $1 \leq i_1 < i_2 < \ldots < i_n$, the quantity
 $\mathbb{P}(F_{i_1} \cap \ldots \cap F_{i_n})$ with arbitrary accuracy, that is, for $k \in \mathbb{N}$, the algorithm yields a dyadic rational number $\beta_k$ such that
    $$\mathbb{P}(F_{i_1} \cap \ldots \cap F_{i_n}) -\beta_k| \leq 2^{-k}.$$
\item[$(3)$] for $n, i \in \mathbb{N}$, a rational number $\epsilon>0$ and $x \in \mathscr{C}_n$, both the relations $x \in O_\epsilon(F_i)$ and $x\in O_\epsilon(F_i^0)$ are recursive in $x, \epsilon, i$ and $n$, that is, there is an algorithm that on input $x, \epsilon, i, n$ yields ``yes'' if indeed $x \in O_\epsilon(F_i)$ and ``no'' if $x \notin O_\epsilon(F_i)$ and similarly for $x\in O_\epsilon(F_i^0)$. 
 \end{itemize}
\end{definition}
Let $\mathscr{F}_0$ be an effective generating sequence and $\mathscr{F}$ the algebra generated by $\mathscr{F}_0$, that is, the class of finite unions of sets of the form
 \begin{eqnarray} \label{sdw3ewa12}
F_{i_1}^{\delta_1} \cap F_{i_2}^{\delta_2}  \ldots \cap F_{i_n}^{\delta_n}
\end{eqnarray}
with $1 \leq i_1 < \ldots < i_n$ in $\mathbb{N}$ and $\delta_i \in \{0, 1\}$. 
Then one can  enumerate the elements of $\mathscr{F}$ in a sequence $T_1, T_2, \ldots, T_n, \ldots$ such that there is algorithm that on input $i$ yields a description of $T_i$ as a finite union of sets of the form (\ref{sdw3ewa12}). In that case, $T_1, T_2, \ldots, T_n, \ldots$ is called a recursive enumeration of $\mathscr{F}$ and we say that $\mathscr{F}$ is effectively generated by $\mathscr{F}_0$.  If $(T_i: i = 1,2, \ldots)$  is  a recursive enumeration of $\mathscr{F}$ effectively generated by the sequence $\mathscr{F_0}$, then there is an algorithm that yields, for each $i$, the quantity $\mathbb{P}(T_i)$ with arbitrary accuracy (\cite[Lemma 4.1]{fo:1}).  
The following notion is defined in \cite{fo:1}.
\begin{definition}
A set $A \subset C[0, 1]$ is of constructive measure 0, if there exist an effectively generated algebra $\mathscr{F}$, a corresponding recursive enumeration $(T_i: i=1,2, \ldots)$ of $\mathscr{F}$ and a recursive function $\phi: \mathbb{N} \times \mathbb{N} \to \mathbb{N}$ such that 
 $$A \subset \cap_{n} \cup_{m} T_{\phi(n, m)}$$ and $\mathbb{P}\left(\cup_m T_{\phi(n, m)}\right)$ converges effectively to 0 for $n \to \infty$. 
\end{definition}
The following result is an important property of complex oscillations \cite[Theorem 4.1]{fo:1}.
\begin{theorem} \label{e3rwdw12}
If $x$ is complex oscillation   and if $A$ is a set of constructive measure 0, then $x \notin A$. 
\end{theorem}
In some cases condition (3) of Definition \ref{sew12dwe23}  ``for $x \in \mathscr{C}_n$, $x \in O_{\epsilon}(F_i)$ is effective in $x, n, \epsilon, i$'' is not easy to verify as  generally it may require to  explicitly find  $f\in F_i$ such that $\|x - f\| < \epsilon$. We shall replace condition (3) with the following condition:
Condition $(3')$:  for each $i$,  both $F_i$ and its complement $F_i^0$ in $C[0, 1]$ have no isolated point and for $n, i \in \mathbb{N}$ and $x \in \mathscr{C}_n$, both the relations $x \in F_i$ and $x\in F_i^0$ are recursive in $x, i$ and $n$.  
This has the implication that if $x\in F_i^{\delta}$ with $\delta \in \{0, 1\}$  and if $x_n$ is a sequence in $C[0, 1]$ that converges to $x$ then $x_n \in F_i^\delta$ for all large $n$. 
Then if $$A = F_{i_1}^{\delta_1} \cap F_{i_2}^{\delta_2}  \ldots \cap F_{i_n}^{\delta_n},\,\,\delta_i \in \{0, 1\},$$
$x \in A$ and  $x_n$ is a sequence in $C[0, 1]$ that converges to $x$, then $x_n A$ for all large $n$. 
This also implies that  for $n, i \in \mathbb{N}$, a rational number $\epsilon>0$ and $x \in \mathscr{C}_n$, both the relations $x \in O_\epsilon(F_i \cap \mathscr{C}_n)$ and $x\in O_\epsilon(F_i^0 \cap \mathscr{C}_n)$ are recursive in $x, \epsilon, i$ and $n$. Indeed, it is enough to enumerate all the elements of $\mathscr{C}_n$ (which is finite) that are in $F_i$ (this can done effectively by Condition $(3')$) and test whether $\|x - y\| < \epsilon$ for each $y \in F_i \cap \mathscr{C}_n$. To show that  it can be done we prove 

 

\begin{theorem} \label{sde3412wsw}
Let $\mathscr{F}_0 = \{F_1, F_2, \ldots, F_n, \ldots\}$ be a sequence of subsets of $C[0, 1]$ satisfying the following conditions:
       \begin{itemize}
       \item[$(a)$]  for each $F \in \mathscr{F}_0$ and  $\epsilon >0$, if  $G \in \{F, F^0, O_\epsilon(F), O_\epsilon(F^0)\}$, then $\mathbb{P}(\bar G) = \mathbb{P}(G).$
 \item[$(b)$] there is an algorithm that yields, for each finite sequence of integers $1 \leq i_1 < \ldots < i_n$, the quantity
 $\mathbb{P}(F_{i_1} \cap \ldots \cap F_{i_n})$ with arbitrary accuracy,
 \item[$(c)$] for each $F \in \mathscr{F}_0$,  both $F$ and its complement $F^0$ in $C[0, 1]$ have no isolated point,
 \item[$(d)$] for $n, i \in \mathbb{N}$ and $x \in \mathscr{C}_n$, both the relations $x \in F_i$ and $x\in F_i^0$ are recursive in $x, i$ and $n$.  
       \end{itemize}
Let $T_1, T_2, \ldots, T_n, \ldots$ be a recursive enumeration of the algebra $\mathscr{F}$ generated by $\mathscr{F}_0$. If  $\phi: \mathbb{N} \times \mathbb{N} \to \mathbb{N}$ is recursive function  such that $$ \mathbb{P}\left(\cup_m T_{\phi(n, m)}\right)$$ converges effectively to 0 for $n \to \infty$, then  the set $$\cap_{n} \cup_{m} T_{\phi(n, m)}$$ does not contain a complex oscillation.               
\end{theorem}

\begin{proof}
 The proof is very similar to the proof of Theorem \ref{e3rwdw12} given in \cite{fo:1}.  As indicated before, there is an algorithm that yields, for each $i$, the quantity  $\mathbb{P}(T_i)$  with arbitrary accuracy.
 
For a string $a\in \{0,1\}^*$ the function $\psi(a)$ defined by $a$ shall also be denoted $a$.  Assume that there is a complex oscillation $x \in \cap_{n} \cup_{m} T_{\phi(n, m)}$. Set $A_{n,m} = T_{\phi(n,m)}$. 
Since $x$ is a complex oscillation, then there is a complex sequence $(x_n)$ in $C[0, 1]$ and a recursive function $f:\mathbb{N} \to \mathbb{Q}^+$ such that $f(n)$ converges effectively to 0 as $n  \to \infty$ with $\|x_n - x\| < f(n)$ for all $n$. We shall assume, without loss of generality, that $f$ is decreasing. Further we can assume without loss of generality that  
 $$ \mathbb{P}(\cup_m T_{\phi(n,m)}) \leq 2^{-n}.$$ 
and that  for each fixed $n$, the sets $A_{n,m}$, $m=1,2,\ldots$ are pairwise disjoint. This can be realised by replacing $A_{n,m}$ by $A_{n,m} \setminus \cup_{\ell < m} A_n,\ell.$ Then we have 
                          $$\sum_m \mathbb{P}(A_{n,m}) \leq 2^{-n}.$$
                          
  We shall construct a recursively enumerable semi-measure $\nu$ so that for given $d>0$, there exists some $n\in \mathbb{N}$ such that $\nu(a) \geq 2^{-n + d}$ for all $a \in \mathscr{C}_n \cap O_{f(n)} (A_{N, M} \cap \mathscr{C}_n)$ for some set $A_{N, M}$ containing $x$, $\nu(x_n) = 2^{-n +d}$. 
Since $x \in A_{N, M}$ and $(x_n)$ converges to $x$, then $x_n \in A_{N, M}$ for all large $n$ (from condition (c) of Theorem~\ref{sde3412wsw} which is condition $(3')$). Then for all large $n$, $x_n \in \mathscr{C}_n \cap O_{f(n)} (A_{N, M} \cap \mathscr{C}_n)$ and hence $\nu(x_n) = 2^{-n +d}$ for all large $n$.  Then  $$K(x_n) \leq -\log_2 \nu(x_n) + D \leq n - d + D$$ for some constant $D$. This contradicts the assumption that $x$ is a complex oscillation for $d$ large. 
 
 To construct $\nu$ we proceed as follows:
 We consider any standard enumeration of $\mathbb{N} \times \mathbb{N}$. Let $(1,1)$ be the first pair. Set $A = A_{2, 1}$ and $\eta = 2^{-2-1} = 2^{-3}$. Since $A \in \{T_1, T_2, \ldots\}$, we can effectively write $A$ as a finite union of sets $B_m$ where each $B_m$ is a finite intersection  $\cap_\ell D_{\ell, m}$ where each $D_{\ell, m}$ or its complement belong to $\mathscr{F}_0$. 

For a given $\epsilon >0$ and $n \in \mathbb{N}$, it is the case that
\begin{eqnarray*}
 \mathscr{C}_n \cap O_\epsilon(A \cap \mathscr{C}_n) & \subset & \mathscr{C}_n \cap O_\epsilon(A)\\
& = &  \mathscr{C}_n \cap O_\epsilon(\cup_m B_m) \\
& =  & \mathscr{C}_n \cap \cup_m  O_\epsilon(B_m)\\
 & = &  \mathscr{C}_n \cap \cup_m  O_\epsilon(\cap D_{\ell,m}) \\
 & \subset & \mathscr{C}_n \cap \left(\cup_m \cap_\ell O_\epsilon(D_{\ell,m})\right).  
 \end{eqnarray*}
 Let
       $$V_1(n, \epsilon) = \cup_m \cap_\ell O_\epsilon(D_{\ell,m}).$$
 Then it is well-known that 
            $$\lim_{n \to \infty} 2^{-n} |V_1(n, \epsilon)| = \mathbb{P}\left(\cup_m \cap_\ell O_\epsilon(D_{\ell,m})\right).$$
 (Here $|V_1(n, \epsilon)|$ denotes the number of elements in the set $V_1(n, \epsilon)$. Moreover,
      $$\lim_{\epsilon \to 0} \mathbb{P}\left(\cup_m \cap_\ell O_\epsilon(D_{\ell,m})\right) = \mathbb{P}\left(\cup_m \cap_\ell (\bar D_{\ell,m})\right) = \mathbb{P}(A).$$

 \end{proof}
 
 \section{Main results}\label{mainresults}

The first main result is related to the pathwise stochastic integral of the sign function with respect to a complex oscillation. The existence of the integral is proven in Mukeru \cite{Mukeru_JOC_2014}. 
  
 \begin{theorem}\label{djkals23}
 For any complex oscillation $\omega$, 
\begin{enumerate}
\item[$(i)$] the sequence of functions $f_n(., \omega): [0, 1] \to \mathbb{R}$, $n =1,2, \ldots,$ defined by
$$f_n(t,\omega) = (2 \epsilon_n)^{-1} \lambda\{s \leq t: |\omega(s)| \leq \epsilon_n\},\,\,\epsilon_n = 2^{-n}$$ 
converges uniformly on $[0, 1]$. The limit 
$$L(t,\omega) = \lim_{n\to \infty} (2 \epsilon_n)^{-1} \lambda\{s \leq t: |\omega(s)| \leq \epsilon_n\}$$ is the {\it (effective) local time} of the complex oscillation $\omega$ at the origin on the interval $[0, t]$. 
\item[$(ii)$] Any complex oscillation $\omega$ satisfies Tanaka's formula:  
\begin{eqnarray} \label{s2qwe2t}
L(t,\omega) = |\omega(t)| - \int_0^t \sign(\omega(s)) d\omega(s).
\end{eqnarray}
\end{enumerate}
\end{theorem}

 The limit 
$$L(t,\omega) = \lim_{n\to \infty} (2 \epsilon_n)^{-1} \lambda\{s \leq t: |\omega(s)| \leq \epsilon_n\}$$ is in fact the local time at the origin of the continuous function  $\omega$ as defined in real analysis.


The proofs of the theorem and its corollary are given in section 5. 

 The following result shows that we can approximate the effective local times of a complex oscillation $\omega$ by the  sum of absolute values of $\omega$ at dyadic points. 
 \begin{theorem}\label{enyanya}
  For any complex oscillation $\omega$, 
  $$L(t, \omega) = 2 \lim_{m\to \infty} \sum_{k \in S_m}|\omega(k/2^m)|$$ 
  where $S_m$ is the subset of $\{1, 2, \ldots, \ell\}$, $\ell = \lfloor t 2^m\rfloor$, defined by
    $$k \in S_m \mbox{ iff } \sign(\omega(k/2^m)) \ne \sign(\omega((k-1)/2^m)).$$
 \end{theorem}
 This implies directly the following result in the classical context.
\begin{corollary} \label{323aarsxza1}
The local times of Brownian motion at the origin satisfy the following property: Almost surely, for all $t\in [0, 1]$,  $$\mathscr{L}(t, X) = 2 \lim_{m\to \infty} \sum_{k \in S_m}|X(k/2^m)|$$  where $S_m$ is the subset of $\{1, 2, \ldots, \ell\}$, $\ell = \lfloor t 2^m\rfloor$, defined by
    $$k \in S_m \mbox{ iff } \sign(\omega(k/2^m)) \ne \sign(\omega((k-1)/2^m)).$$
\end{corollary}
 The proof of Theorem \ref{enyanya} is given in section 6.

 \section{Pathwise stochastic integral with respect to complex oscillations} \label{mushana}
 
Local times of Brownian motion have been obtained as stochastic integrals with respect to Brownian motion. Pathwise stochastic integral with respect to  Martin-L\"of random Brownian motion or complex oscillation  was introduced in \cite{Mukeru_JOC_2014}.

Denote by $X: [0, 1]\times C[0, 1] \to \mathbb{R}$ the standard one-dimensional Brownian motion on $[0, 1]$. For any $t \in [0, 1]$ and $\omega \in C[0, 1]$, $X(t, \omega) = \omega(t)$. 


 We will consider real functions $f:[0, 1] \times C[0, 1] \to \mathbb{R}$, $(t,\omega) \mapsto f(t,\omega) = \varphi(\omega(t))$ where $\varphi: \mathbb{R} \to \mathbb{R}$ is a Borel function. We investigate such  functions $f$ satisfying the following three conditions: \\ 

\begin{itemize}
\item
\begin{eqnarray}\label{sdw345ers}
\mathbb{P}\left[\omega \in \Omega:\int_0^1 f^2(t,\omega) dt <\infty\right] = 1. 
\end{eqnarray} 
\item
The sequence of simple functions $f_n: [0, 1] \times C[0, 1]$, $n = 1, 2, \ldots$ defined by $$f_n(t, \omega) = f(k/2^n, \omega) \mbox{ for } k/2^{n} \leq t \leq (k+1)/2^{n}, \,\, k = 0, 1, 2, \ldots, 2^n-1$$ converges to $f$ in the  sense that for all large $n\geq 1$, 
 \begin{eqnarray} \label{eqnkaa2}
\mathbb{P}\left[\omega: \int_0^1 \left(f(s,\omega)-f_n(s,\omega)\right)^2 ds > 2^{-n }\right] \leq 2^{-n}
 \end{eqnarray}\\
\item
The function $f$ is computable in the sense that there exists an algorithm that on input $(t_i, x_i)$ with $t_i$ rational and $x_i \in \mathscr{S}$ yields an approximation of $f(t_i, x_i)$ with arbitrary accuracy. (We recall that $\mathscr{S}$ is the class of piecewise linear functions in $C[0, 1]$ with rational coordinates for its points of non-differentiability.)
 \end{itemize}

Condition (\ref{eqnkaa2}) implies that 
  \begin{eqnarray} \label{eqnwe23de}
\mathbb{P}\left[\omega: \int_0^1 \left(f_n(s,\omega)-f_{n-1}(s,\omega)\right)^2 ds > 2^{-n+4}\right] \leq 3\times 2^{-n}.
 \end{eqnarray}
Indeed, using the $L^2$-norm, 
  \begin{eqnarray*}
 \left(\int_0^1 |f_{n}(s)- f_{n-1}(s)|^2 ds\right)^{1/2} & =  & \left(\int_0^1 |f_{n}(s)- f(s) + f(s)- f_{n-1}(s)|^2 ds\right)^{1/2}\\
&  \leq & \left(\int_0^1 (f_{n}(s)- f(s))^2 ds \right)^{1/2} \\
&& + \left(\int_0^1 (f(s)- f_{n-1}(s)|^2 ds\right)^{1/2}. 
\end{eqnarray*}
Set $$T_n = \int_0^1 (f_{n}(s)- f(s))^2 ds.$$ The the immediate implication
$(a+b\geq c)\longrightarrow  (a\geq c/2 \mbox{ or } b\geq c/2)$ yields
 \begin{eqnarray*}
 \mathbb{P}\left[\int_0^1 |f_{n}(t) - f_{n-1}(t)|^2 dt > 2^{-n+4}\right] & = & \mathbb{P}\left[\left(\int_0^1 |f_{n}(t) - f_{n-1}(t)|^2 dt\right)^{1/2}> 2^{(-n+4)/2}\right] \\
 &\leq & \mathbb{P}[T_n^{1/2} + T_{n-1}^{1/2} > 2^{(-n+4)/2}]\\
 &\leq & \mathbb{P}[T_n^{1/2}> 2^{(-n+2)/2} \mbox{ or } T_{n-1}^{1/2}> 2^{(-N+2)/2}]\\
 &\leq & \mathbb{P}[T_n^{1/2}> 2^{(-n+2)/2}] + \mathbb{P}[T_{n-1}^{1/2}> 2^{(-n+2)/2}]\\
 &=& \mathbb{P}[T_n> 2^{(-n+2)}] + \mathbb{P}[T_{n-1}> 2^{(-n+2)}]\\
 & \leq & \mathbb{P}[T_n> 2^{-N}] + \mathbb{P}[T_{n-1}> 2^{-n+1}]\\
 & \leq  & 2^{-n} + 2^{-n-1}\\
 &\leq & 3 \times 2^{-n}. 
 \end{eqnarray*}
For each integer $n \geq 1$ and $\omega \in C[0, 1]$, the function 
 $I(f, ., \omega): [0, 1] \to \mathbb{R}$ defined by
\begin{eqnarray*}
I(f_n, t, \omega) & = & \sum_{k=1}^{\ell} f((k-1) 2^{-n},\omega) (\omega(k 2^{-n}) - \omega((k-1) 2^{-n})) \\
&& + f(\ell 2^{-n},\omega) (\omega(t) - \omega(\ell 2^{-n})) ,\,\,  \ell =\lfloor  2^n t\rfloor. 
\end{eqnarray*}
 is called the pathwise stochastic integral of $f_n$ with respect to the path $\omega$.  The quantity 
$I(f_n, t, \omega)$ is denoted $\int_0^t f_n(s, \omega) d\omega(s)$. \\
It is shown in \cite[Theorem 2]{Mukeru_JOC_2014} that if a function $f: [0, 1] \times C[0, 1]$ satisfies all the three conditions  as stated above, then for any {\it complex oscillation} $\omega$, the sequence of functions $I(f_n,., \omega)$ converges uniformly on $[0, 1]$. The limit function is called the pathwise stochastic integral of $f$ with respect to $\omega$. It is denoted
      $$I(f, t, \omega) := \int_0^t f(s, \omega) d\omega(s) := \lim_{n\to \infty} \int_0^t f_n(s, \omega) d\omega(s).$$ 
Moreover the pathwise stochastic integral of $f$ does not depend on the particular approximation sequence of simple functions $(f_n)$ in the sense that if $(g_n)$ is any other sequence of simple functions satisfying these  conditions, then 
    for any complex oscillation $\omega$,
  $$\lim_{n \to \infty} \int_0^t f_n(s, \omega) d\omega(s) = \lim_{n \to \infty} \int_0^t g_n(s, \omega) d\omega(s)$$ 
  uniformly in $t \in [0, 1]$. 
It is shown in the proof of Theorem 2  in \cite{Mukeru_JOC_2014} that condition (\ref{eqnwe23de}) implies that 
there exists a constant $K$ such that for any complex oscillation $\omega$ there exists an integer $n_0$ such that for all $m, n \geq 0$ and for any $t \in [0, 1]$, 
\begin{eqnarray} \label{mushana1}
\left|\int_0^t (f_m(s, \omega) - f_n(s, \omega)) d\omega(s)\right| \leq n 2^{-n/2} K.
\end{eqnarray}
We will make use of the following property of pathwise stochastic integration.
\begin{theorem}\label{sddfsdsse3ed}
Let $f_N: [0, 1] \times C[0, 1] \to \mathbb{R}$, $N = 1,2, \ldots$ be a sequence of functions such that there exists a sequence of functions $\varphi_N: \mathbb{R} \to \mathbb{R}$, $N = 1,2,\ldots$ satisfying
 $f(t, \omega) = \varphi(\omega(t))$ for all $t\in [0,1]$ and $\omega \in C[0, 1]$. Assume that the following conditions are satisfied: \begin{itemize}
 \item[$(1)$] each function $f_N$ satisfies condition (\ref{sdw345ers}), 
\item[$(2)$] for each $f_n$, the sequence $(f_{N,n}), n = 1,2, \ldots$ defined by 
$$f_{N,n}(t, \omega) = f_N(k/2^n, \omega) \mbox{ for } k/2^{n} \leq t < (k+1)/2^{n}, \,\, k = 0, 1, 2, \ldots, 2^n-1$$ converges to $f_N$ in the sense that for all large $n\geq 1$, 
  \begin{eqnarray} \label{eqnkaa1dk3453}
\mathbb{P}\left[\omega: \int_0^1 \left(f_{N,n}(s,\omega)-f_N(s,\omega)\right)^2 ds > 2^{-n }\right] \leq 2^{-n},
 \end{eqnarray}
\item[$(3)$] the sequence $(f_N)$ is uniformly computable in the sense that there exists an algorithm that, on input $N \geq 1$, $t_i$ rational in $[0, 1]$ and $x_i \in \mathscr{S}$, yields $f_N(t_i, x_i)$ with arbitrary accuracy,
\item[$(4)$] the sequence $(f_N)$ converges to a function $f$ in the sense that \begin{eqnarray} \label{esad23ds}
P\left[\omega: \int_0^1 |f_{N}(t,\omega) - f(t,\omega)|^2 dt > 2^{-N}\right] \leq \, 2^{-N },\,\forall N.
\end{eqnarray}
\end{itemize}
 Then for any complex oscillation $\omega$, the pathwise stochastic integral of $f$ with respect to $\omega$ exists and 
$$\lim_{N\to \infty} \int_0^t f_N(s, \omega(s)) d\omega(s) = \int_0^t f(s, \omega(s))d\omega(s)$$ uniformly in $ t\in [0, 1]$.  
\end{theorem}

\begin{proof}
Denote, for fixed $N,\,n$, 
\begin{eqnarray*}
A & = &\int_0^1 |f_N(s) - f(s)|^2 ds,\\
 B_{n} &= &\int_0^1 |f_{N,n}(s) - f_N(s)|^2 ds,\\
 C  & =& \int_0^1 |f_{N,N}(s) - f(s)|^2 ds,
 \end{eqnarray*}
Then $\mathbb{P}(A >2^{-N}) \leq 2^{-N}$, $\mathbb{P}(B_{n} > 2^{-n}) \leq 2^{-n}$ and 
$\mathbb{P}(C > 2^{-N+4}) \leq 2^{-N+1}.$  Then as discussed before, Theorem 2  in \cite{Mukeru_JOC_2014} implies, for any complex oscillation $\omega$, both the existence of the pathwise stochastic integral of $f$ with respect to $\omega$ and the identity 
 \begin{eqnarray}\label{31242waa}
\int_0^t f(s, \omega(s))d\omega(s)= \lim_{N\to \infty} \int_0^t f_{N,N}(s, \omega(s)) d\omega(s)
\end{eqnarray}
 uniformly in $t \in [0, 1]$. \\
Similarly,
 $$\mathbb{P}\left[\int_0^1 (f_{N,N}(s) - f_{N,N-1}(s))^2 ds > 2^{-N+4}\right]\leq 2^{-N} + 2^{-N +1} = 3\times 2^{-N}.$$ 
Then, again Theorem 2  in \cite{Mukeru_JOC_2014} implies that, there exists a constant $K>0$ such that for any complex oscillation $\omega$, there exists an integer $n_0 > 0$ such that for all $m, N > n_0$ and for all $t \in [0, 1]$,
 $$\left|\int_0^t (f_{N, m}(s, \omega) - f_{N, N}(s, \omega)) d\omega(s)\right| \leq N 2^{-N/2} K.$$
Now taking the limit for $m \to \infty$  and $N$ fixed yields,
 $$\left|\int_0^t (f_{N}(s, \omega) - f_{N, N}(s, \omega)) d\omega(s)\right| \leq N 2^{-N/2} K$$ since 
    $$\lim_{m\to \infty} \int_0^t f_{N, m}(s, \omega) d\omega(s) = \int_0^t f_N(s, \omega) d\omega(s)$$ uniformly in $t \in [0, 1]$. 
It follows  that 
  \begin{eqnarray}\label{rsse34edsw}
\lim_{N\to \infty} \int_0^{t} f_{N,N}(s,\omega)d\omega(s) = \lim_{N\to \infty} \int_0^t f_{N}(s,\omega)d\omega(s)
\end{eqnarray}
uniformly in $t \in [0, 1]$. 
Relations (\ref{31242waa}) and (\ref{rsse34edsw}) conclude the proof. 
\end{proof}

\section{Effective local times of complex oscillations}  \label{sdsde312wa}

\subsection{It\^o's lemma for complex oscillations}
 In the next section, we will make use of the following effective version of Ito's lemma \cite{Mukeru_JOC_2014}.
 \begin{theorem} \label{qaqsaersds}
Let $f :[0, 1] \times \mathbb{R}  \to \mathbb{R}$, $(t,x) \mapsto f(t, x)$ be a $C^2$-function such that 
\begin{eqnarray} \label{e234wwaq}
\sup_{t \leq 1} \left(\mathbb{E} \left|\frac{\partial^2 f}{\partial x^2}(t, X(t))\right|\right) < \infty.
\end{eqnarray}
Further assume  that $\frac{\partial^2 f}{\partial x^2}$  
is computable on $[0, 1] \times \mathbb{R}$. 
Then  for any complex oscillation $\omega$ and any $0 \leq t \leq 1$, It\^o's formula holds: 
\begin{eqnarray*}
f(t, \omega(t)) = f(0,0) + \int_0^t \frac{\partial f}{\partial t}(s, \omega(s)) ds + \int_0^t \frac{\partial f}{\partial x}(s, \omega(s)) d\omega(s) + \int_0^t {\scriptstyle\frac{1}{2}} \frac{\partial^2 f}{\partial x^2}(s, \omega(s)) ds.
\end{eqnarray*}
 \end{theorem}

\subsection{Proof of Theorem \ref{djkals23}}

 
For the existence of effective local times and Tanaka's formula, we use effective versions of some constructions given in  
\cite[pp 127-142]{Chung_William_83}. 
  For a given rational number $\epsilon >0$, consider the function $f_\epsilon:\mathbb{R} \to \mathbb{R}$ given by:
   \begin{eqnarray*}
  f_\epsilon (x) =  \left\{\begin{array}{ccc}
   0 & \mbox{ if } x\leq -\epsilon \\
   (x+\epsilon)^2/4\epsilon  & \mbox{ if } |x| < \epsilon \\
    x & \mbox{ otherwise}.
    \end{array}
    \right.
  \end{eqnarray*}
 Then 
  \begin{eqnarray*}
  f'_\epsilon (x) =  \left\{\begin{array}{ccc}
   0 & \mbox{ if } x\leq -\epsilon \\
   (x+\epsilon)/2\epsilon  & \mbox{ if } |x| < \epsilon \\
    1 & \mbox{ otherwise}
    \end{array}
    \right.
  \end{eqnarray*}
  and \begin{eqnarray*}
  f''_\epsilon (x) =  \left\{\begin{array}{ccc}
      1/2\epsilon  & \mbox{ if } |x| < \epsilon \\
    0 & \mbox{ if }  |x| > \epsilon.
    \end{array}
    \right.
  \end{eqnarray*}  
  Set $f''(x) = 0$ for $x = \pm \epsilon$. Since $f''$ is not continuous, It\^o's formula is not applicable here. To overcome this, we first approximate $f_\epsilon$ by a sequence of convolution products $$g_n(x) = \phi_n * f_\epsilon(x) = \int_{-\infty}^{\infty} \phi_n(z) f_{\epsilon}(x - z) dz, \forall n\geq 1 $$ where $(\phi_n)$ is a uniformly computable sequence of $C^\infty$-functions with support  in $[-2^{-n}, 2^{-n}]$ and such that $\int_{-\infty}^{+\infty} \phi_n(x) dx = 1$ for all $n$. 
We have that 
 $$g'_n(x) = \int_{-\infty}^{\infty}  f'_{\epsilon}(x - z) \phi_n(z) dz,\,\, \,\, g''_n(x) = \int_{-\infty}^{\infty}  f''_{\epsilon}(x - z) \phi_n(z) dz,\,\,x\in \mathbb{R}.$$ Moreover, the sequences $(g_n)$, $(g'_n)$ and $(g''_n)$ are uniformly computable in $n$. Also It\^o's formula is applicable to each function $g_n$ ($g_n$ satisfies condition  (\ref{e234wwaq}) of Theorem \ref{qaqsaersds}). Then for any complex oscillation $\omega$, 
  \begin{eqnarray*}
g_n\omega(t))- g_n(0) = \int_0^t g'_n(\omega(s)) d\omega(s) + \int_0^t {\scriptstyle\frac{1}{2}} g''_n(\omega(s)) ds.
\end{eqnarray*}
For $n \to \infty$, $g_n(x) \to f_\epsilon(x)$ and $g_n'(x) \to f'_\epsilon(x)$ for any $x$ and $g''_n(x) \to f''_\epsilon(x)$ for $x \ne \pm \epsilon.$
Clearly, for any $\omega \in C[0, 1]$, $$\int_0^1 (g'_n(\omega(s)) - f'_\epsilon(\omega(s)))^2 ds \leq (1/4\epsilon^2)\, 2^{-2n}$$
from which it follows, by Theorem \ref{sddfsdsse3ed}, that  
 for any complex oscillation, uniformly on $[0, 1]$,
  \begin{eqnarray} \label{sds23sds}
\lim_{n\to \infty} \int_0^t g'_n(\omega(s)) d\omega(s)) = \int_0^t f'_\epsilon(\omega(s))d\omega(s).
\end{eqnarray}
We also have that, for any complex oscillation $\omega$, 
   \begin{eqnarray} \label{11wwe3e345e}
\lim_{n\to \infty} \int_0^t g''_n(\omega(s)) ds = \int_0^t f''_\epsilon(\omega(s))ds
\end{eqnarray}
uniformly in $t \in [0, 1]$. 
   Indeed, for $\epsilon >0$ fixed, since $g''_n(\omega(t))\to f''_{\epsilon}(\omega(t))$ as $n\to \infty$ for $\omega(t) \ne \pm \epsilon$, and $|g''_n| \leq 1/2\epsilon$, the result follows from by the bounded convergence theorem and the that the level set $$Z_a = \{t\in [0, 1]: \omega(t) = a\}$$  of any complex oscillation $\omega$ has Lebesgue measure 0 for any computable real number $a$. (It is shown in \cite{fo:5}  that the zero set $Z_0$ of a complex oscillation has dimension $\leq 1/2$ and in particular $\lambda(Z_0) = 0$. The proof extends easily to all computable real numbers $a$.)
   
    Then for any complex oscillation $\omega$ and any $t \leq 1$, (\ref{sds23sds}) and (\ref{11wwe3e345e}) gives 
\begin{eqnarray*}
 \lim_{n \to\infty} \int_0^t g'_n(\omega(s) d\omega(s) + {\scriptstyle\frac{1}{2}} \int_0^t g''_n(\omega(s)) ds  \\
 =\int_0^t f'_\epsilon \omega(s) d\omega(s) + {\scriptstyle\frac{1}{2}} \int_0^t f''_\epsilon(\omega(s)) ds.
 \end{eqnarray*}
Therefore 
   $\lim_{n \to \infty} g_n(x) = f_\epsilon(x)$, it follows using It\^o's formula on $g_n$ that 
     \begin{eqnarray} \label{eqanwew23}
f_\epsilon(\omega(t)) - f_\epsilon(0) = \int_0^t f'_\epsilon (\omega(s)) d\omega(s) + {\scriptstyle\frac{1}{2}} \int_0^t f''_\epsilon(\omega(s)) ds.
\end{eqnarray}
for any complex oscillation $\omega$. 
This is Ito's formula for the function $f_\epsilon$. 
 We can now consider the limit as $\epsilon \to 0$ in (\ref{eqanwew23}).
 
1) First, we have that for $\epsilon \to 0$,  $f_\epsilon(x) \to x^+ =  \max\{x, 0\}$ and then $$f_\epsilon(\omega(t)) \to \omega^+(t) = 1_{[0,+\infty)}(\omega(t)).$$

2) For the first derivative,
    \begin{eqnarray*}
  \lim_{\epsilon \to 0} f'_\epsilon (x) =  \left\{\begin{array}{ccc}
   0 & \mbox{ if } x < 0\\
   1 & \mbox{ if } x > 0 \\
    1/2 &  \mbox{ if } x = 0.
    \end{array}
    \right.
     \end{eqnarray*}
  We now take $\epsilon_n = 2^{-n}$, $ n\geq 1$ and show that 
 \begin{eqnarray} \label{eqawwee34}
  \mathbb{P}\left[\int_0^1 \left(f'_{\epsilon_n} (\omega(t)) - 1_{[0,+\infty)}(\omega(t))\right)^2 dt > \epsilon_n^{1/2}\right] \leq  \frac{\epsilon_n^{1/2}}{3 \sqrt{2\pi}}\,\,\,\forall n \geq 1.
 \end{eqnarray}
This relation together with Theorem \ref{sddfsdsse3ed} yield 
$$\lim_{n \to \infty} \int_0^t f'_{\epsilon_n}(\omega(s))d\omega(s) = \int_0^t 1_{[0, \infty)}(\omega(s))d\omega(s)$$ for any complex oscillation $\omega$ uniformly in $t$. \\
To obtain  (\ref{eqawwee34}),
 \begin{eqnarray*}
 \mathbb{E}\left[\int_0^1 \left((f'_\epsilon - 1_{[0,+\infty)})(X(t)\right)^2 dt \right] & = & \int_0^1 \mathbb{E}\left(f'_\epsilon  - 1_{[0,+\infty)})(X(t))\right)^2 dt 
 \end{eqnarray*}
 and since $X(t)$ is normally distributed with mean 0 and variance $t$,
\begin{eqnarray*}
 \mathbb{E}\left((f'_{\epsilon_n} - 1_{[0,+\infty)})(X(t))\right)^2 & = & \int_{-\infty}^{+\infty}\left(f'_{\epsilon_n}(x)-1_{[0,+\infty)}(x)\right)^2 e^{-x^2/2t}/\sqrt{2 \pi t}dx\\
 &\leq & \frac{1}{(2 {\epsilon_n})^2 \sqrt{2 \pi t}}\left(\int_{-{\epsilon_n}}^0 (x+{\epsilon_n})^2 dx + \int_0^{\epsilon_n} (x-\epsilon_n)^2 dx\right)\\
&= &\epsilon_n/(6\sqrt{2 \pi t}).
 \end{eqnarray*}
Therefore
  $$\mathbb{E}\left[\int_0^1 \left((f'_{\epsilon_n} - 1_{[0,+\infty)})(X(t))\right)^2 dt \right] \leq \epsilon_n/(3 \sqrt{2\pi}).$$
This implies (\ref{eqawwee34}) by  Chebyshev's inequality. 

3)  Since 
 \begin{eqnarray*}
  f''_\epsilon (\omega(s)) =  \left\{\begin{array}{ccc}
      1/2\epsilon  & \mbox{ if } |\omega(s)| < \epsilon \\
    0 & \mbox{ if }  |\omega(s)| > \epsilon
    \end{array}
    \right.
  \end{eqnarray*} 
and $\lambda\{t\leq 1: \omega(t) = \epsilon\} = 0$ for any complex oscillation $\omega$, it follows that 
  $$\int_0^t f''_\epsilon (\omega(s)) ds = \frac{1}{2\epsilon} \int_0^t 1_{\{s: |\omega(s)| < \epsilon\}}(s) ds.$$
 
Therefore, taking $\epsilon = \epsilon_n = 2^{-n}$, relation (\ref{eqanwew23}) yields,
 \begin{eqnarray} \label{31q242wsd}
L(t,\omega) & = & \lim_{n \to +\infty} \frac{1}{2\epsilon_n} \int_0^t 1_{\{s: |\omega(s)| < \epsilon_n\}}(s)\,ds\nonumber\\
& = & 2\left[\omega^+(t) - \int_0^t 1_{[0, \infty)}(\omega(s))d\omega(s)\right]
 \end{eqnarray}
 for any complex oscillation $\omega$ uniformly in $t\in [0, 1]$. 
 The existence of the limit follows from the existence of the pathwise stochastic integral of $\int_0^t 1_{[0, \infty)}(\omega(s))d\omega(s)$. 
 
If we replace function $f_\epsilon$ by 
 \begin{eqnarray*}
  h_\epsilon (x) =  \left\{\begin{array}{ccc}
   -x & \mbox{ if } x\leq -\epsilon \\
   (-x+\epsilon)^2/4\epsilon  & \mbox{ if } |x| < \epsilon \\
    0 & \mbox{ otherwise}.
    \end{array}
    \right.
  \end{eqnarray*}
we obtain that, by the same calculations,  
 \begin{eqnarray}\label{sdw1q242wsd}
L(t,\omega) = 2\left[\omega^-(t) + \int_0^t 1_{(-\infty, 0]}(\omega(s))d\omega(s)\right]
 \end{eqnarray}
where $x^- = \max\{0, -x\}.$ Adding (\ref{31q242wsd}) and (\ref{sdw1q242wsd}) yields Tanaka's formula 
   \begin{eqnarray}\label{sdwssa1q2942wsd}
L(t,\omega) =|\omega(t)| - \int_0^t \mbox{sign}(\omega(s))\,d\omega(s)
 \end{eqnarray}
 for every complex oscillation $\omega$. Note that for any complex oscillation,
    $$\int_0^t 1_{[0, \infty)}(\omega(s))d\omega(s) = \int_0^t 1_{(0, \infty)}(\omega(s))d\omega(s)$$ since
    $\omega(k/2^n) \ne 0$ for all integers $n, k \ne 0$. 
This concludes the proof. \hfill \qed

 

\section{Discrete  approximations of local times}  \label{e132ssw23}

In this section, we obtain a very simple approximation of local times of a  complex oscillation $\omega$ depending only on a finite number of values of $\omega$ at dyadic points (Theorem \ref{enyanya}). This is probably the simplest representation of local times known to the authors. 
\paragraph{Proof of Theorem \ref{enyanya}} 

Consider the Haar system in $L^2[0, 1]$: 
 $$e_0 = 1, \, e_1 = \chi([0, 1/2)) - \chi([1/2, 1)$$
and 
  $$e_{jn} = 2^{j/2}\left(\chi[n 2^{-j}, n 2^{-j} + 2^{-(j+1)}) - \chi[n 2^{-j} + 2^{-(j+1)}, (n+1) 2^{-j})\right)$$
 where $n , j$ are integers such that $j \geq 1$ and $0\leq n < 2^j$. Here $\chi(I)$ is the indicator function of interval $I$. Let $\Delta_0(t), \Delta_1(t), \Delta_{jn}(t)$ be the Schauder functions on $[0, 1]$ obtained by integrating $e_0, e_1, e_{jn}$ from $0$ to $t$: 
   $$\Delta_h(t) = \int_0^t e_h(s) ds,\,\,\,h = 0, 1, jn: n < 2^j \mbox{ in } \mathbb{N}. $$
 Any path $\omega \in C[0, 1]$ can be approximated by a finite linear combination of Schauder functions on $[0, 1]$ as follows:
\begin{eqnarray*}
 \omega_m(t) =  \xi_0 \Delta_0(t) + \xi_1 \Delta_1(t) + \sum_{1 \leq j \leq m} \sum_{n <2^j} \xi_{jn}\Delta_{jn}(t) 
 \end{eqnarray*}
where the coefficients $\xi_0, \xi_1, \xi_{jn}$ are given by 
 $$ \xi_0 = \omega(1),\,\, \xi_1 = 2\,\omega(1/2) - \omega(1), \,\, \xi_{jn} = 2^{j/2} (2\, \omega(t_{jn}) - \omega(t_{jn} + \delta_j) - \omega(t_{jn} - \delta_j)),\, $$ where 
 $$t_{jn} = (2n + 1)2^{-(j+1)},\,\, \delta_j = 2^{-(j +1)}.$$ 
 It is shown in \cite{fo:2} that for any complex oscillation $\omega$,  there exists $m_0 \in \mathbb{N}$ (depending on $\omega$) such that for all $m < m_0$, 
 \begin{eqnarray} \label{q1sdsds234}
\|\omega -\omega_m\| \leq C \frac{\sqrt{m}}{2^{m/2}}.
\end{eqnarray}
 Clearly, 
 $$\omega_m(t) = \omega(t) \mbox{ for } t = k 2^{-m}, \,\,k = 1, 2, \ldots, 2^m,$$
(that is, $\omega_m$ is an exact approximation of $\omega$ at dyadic rationals $k 2^{-m}$, for  $k = 1, 2, \ldots, 2^m$).  
 We can now prove Theorem \ref{enyanya}. 

Let $f:[0, 1] \times C[0, ]$:  $(t,\omega) \mapsto f(t,\omega) = \sign(\omega(t))$. Consider the sequence $f_n: [0,1] \times C[0, 1]$, $n=1,2, \ldots,$ defined by
 $$f_n(t,\omega) = \sign(\omega((k-1) 2^{-n})),\,\, \mbox{ for }k 2^{-n} \leq t < k 2^{-n},\, k=1,2, \ldots,2^n.$$
It is proven in \cite{Mukeru_JOC_2014}(proof of Theorem 3) that 
       \begin{eqnarray*}
\mathbb{P}\left[\omega: \int_0^1 \left(f_n(s,\omega)-f_{n-1}(s,\omega)\right)^2 ds > 2^{-n/4}\right] \leq 8\times 2^{-n/4}.
 \end{eqnarray*}
which, as discussed in section \ref{mushana} (see relation (\ref{mushana1})), 
implies that for any complex oscillation $\omega$, there exists $n_0 \in \mathbb{N}$ such that for all $M, m \geq n_0$, 
\begin{eqnarray}\label{musii1223}
\sup_{0\leq t\leq 1}\left|\int_0^t f_M(s,\omega) d\omega(s) - \int_0^t f_m(s,\omega) d\omega(s) \right|
\leq K 2^{-m/2} m. 
\end{eqnarray}
In particular, 
 \begin{eqnarray}\label{musii1223}
\sup_{0\leq t\leq 1}\left|\int_0^t f(s,\omega) d\omega(s) - \int_0^t f_m(s,\omega) d\omega(s) \right|
\leq K 2^{-m/2} m, \forall m\geq n_0. 
\end{eqnarray}
By definition,
 \begin{eqnarray} \label{32342ere1}
\int_0^t f_m(s, \omega) d\omega(s) & = & \sum_{k=1}^{\ell} \mbox{sign}(\omega((k-1)/2^m))\left[\omega(k/2^m) - \omega((k-1)/2^m)\right]\nonumber\\
&&+ \mbox{ sign}(\omega(\ell/2^m))\left[\omega(t) - \omega(\ell/2^m)\right],\,\,\,\, \ell = \lfloor t 2^m \rfloor. \end{eqnarray} 
For each $0 \leq t \leq 1$, define $\int_0^t f_m(s, \omega_m) d\omega_m(s)$ by
\begin{eqnarray} \label{sdsw3ed2q}
\int_0^t f_m(s, \omega_m) d\omega_m(s) & = & \sum_{k=1}^{\ell} \mbox{sign}(\omega_m((k-1)/2^m))\left[\omega_m(k/2^m) - \omega_m((k-1)/2^m)\right]\nonumber\\
&&+ \mbox{ sign}(\omega_m(\ell/2^m))\left[\omega_m(t) - \omega_m(\ell/2^m)\right],\,\,\,\, \ell = \lfloor t 2^m \rfloor. \end{eqnarray} 
Since $\omega_m(k/2^m) = \omega(k/2^m)$ for all $k=1,2, \ldots, 2^m$, then (\ref{32342ere1}) and (\ref{sdsw3ed2q}) yield 
\begin{eqnarray} \label{ssd23bnh}
\left|\int_0^t f_m(s, \omega) d\omega(s) - \int_0^t f_m(s, \omega_m) d\omega_m(s)\right| &=&|\omega(t) - \omega_m(t)| \leq  C\sqrt{m}/2^{m/2}
\end{eqnarray}
for all large numbers $m > m_0$ (by  (\ref{q1sdsds234})).  Then relations (\ref{musii1223}) and (\ref{ssd23bnh}) imply that for all $m\geq \max\{m_0, n_0\}$, 
$$\sup_{0\leq t\leq 1}\left|\int_0^t f(s,\omega) d\omega(s) - \int_0^t f_m(s,\omega_m) d\omega_m(s) \right|
\leq K 2^{-m/2} m + C\sqrt{m}/2^{m/2}.$$
Then Tanaka's formula (\ref{s2qwe2t})
yields
   \begin{eqnarray*}
\sup_{0\leq t\leq 1}\left|L(t, \omega) - \left(|\omega_m(t)|- \int_0^t f_m(s, \omega_m) d\omega_m(s)\right) \right| &\leq &  K 2^{-m/2} m \\
&&+ 2 C\sqrt{m}/2^{m/2}
\end{eqnarray*} 
since $|\omega(t) - \omega_m(t)|\leq C\sqrt{m}/2^{m/2}$.
It follows that uniformly in $t$, 
 \begin{eqnarray} \label{qwwe2w2w}
L(t,\omega) & = & \lim_{m\to \infty}\left(|\omega_m(t)| - \int_0^t f_m(s, \omega_m) d\omega_m(s)\right)\nonumber\\
  & = & |\omega(t)| - \lim_{m\to \infty}\int_0^t f_m(s, \omega_m) d\omega_m(s).
\end{eqnarray}
Denote for simplification purpose
  $$x_k = \omega(k/2^m), k = 1, 2, \ldots, 2^m.$$
By the continuity of $\omega$,  $$\omega(t) = \lim_{m\to \infty} \omega(\ell/2^m),\,\,\ell = \lfloor t 2^m \rfloor$$
Clearly, from (\ref{qwwe2w2w}), 
\begin{eqnarray*}
L(t,\omega,a) & = & \lim_{m\to \infty} \left(|\omega(\ell/2^m)-a| -|a|- \sum_{k=1}^\ell (\mbox{sign}(x_{k-1}-a)) (x_k - x_{k-1})\right)\\
 & = & \lim_{m\to \infty} \left(|x_\ell-a|-|a| - \sum_{k=1}^\ell (\mbox{sign}(x_{k-1}-a)) (x_k - x_{k-1})\right).\\
\end{eqnarray*}
For any sequence of real numbers $t_0, t_1, \ldots, t_n$, 
$$|t_n| - \sum_{k=1}^n \mbox{sign}(t_{k-1}) (t_k - t_{k-1}) = t_0 + \sum_{k\in S} |2 t_k|$$ where 
$$S = \{k\in\{1, 2, \ldots, n\}: \mbox{sign}(t_k) \ne \mbox{sign}(t_{k-1})\}.$$
It follows that 
\begin{eqnarray*}
L(t,\omega) & =& \lim_{m\to \infty} 2 \sum_{k\in S_\ell} |\omega(k 2^{-m})| 
\end{eqnarray*}
where $$S_\ell= \{k\in\{1, 2, \ldots, \ell\}: \mbox{sign}(\omega(k/2^{m})) \ne \mbox{sign}(\omega((k-1)/2^{m}))\}. \hfill \qed$$ 




\begin{thebibliography}{99}

\bibitem{Allen_Bienvenu_Slaman} Allen,K., Bienvenu, L. and Slaman, T.A., On zeros of Martin-L\"of random Brownian motion,  {\em Journal of Logic and Analysis.} {\bf 6(9)} (2014), 1-34.


\bibitem{ap:1} Asarin, E. A. and Prokovskii, A. V., Use of the Kolmogorov complexity in analysing control system dynamics, {\it Automat. Remote Control} {\bf 47} (1986), 21-28. 

\bibitem{Berman_1972} Berman, S.M., Gaussian sample functions: uniform dimension and H\"older conditions nowhere. {\it Nagoya Math. J.} {\bf 46} (1972), 63-86.

\bibitem{Bichteler_1981} Bichteler, K. Stochastic integration and $L^p$-theory of stochastic integration, {\em Ann. Prob.} {\bf 9} (1981), 48-89.

\bibitem{ch:2} Chaitin, G. A., {\em Algorithmic information theory}, Cambridge University Press, 1987.

\bibitem{Downey_Hirschfeldt} Downey, R. G. and Hirschfeldt, D.R., {\em Algorithmic randomness and complexity}. Theory and Applications of Computability. Springer, New York, 2010.

\bibitem{Chung_William_83} Chung, K.L. and Williams, R.J., {\it Introduction to stochastic integration}. Birk\"auser, 1983.


\bibitem{Davie_Fouche} Davie, G. and Fouch\'e W.L., On the computability of a construction of Brownian
motion, {\it Math. Structures Comput. Sci}., {\bf 23(6)} (2013), 257 1265.
 
\bibitem{Follmer} F\"ollmer, H., Calculs d'It\^o sans probabilit\'es, {\em Sem. Probab.} {\bf 15} (1981), 143-150.

\bibitem{fo:1} Fouch\'{e}, W. L., Arithmetical representations of Brownian motion I, {\em J. Symb. Logic} {\bf 65} (2000), 421-442.

\bibitem{fo:2} Fouch\'e, W. L., The descriptive complexity of Brownian motion, {\em Adv. Math.} {\bf 155} (2000), 317-343.

\bibitem{fo:4} Fouch\'e, W. L., Dynamics of a generic Brownian motion : Recursive aspects, in: From G\"odel to Einstein: Computability between Logic and Physics, {\em Theoret. Comput. Sci.} {\bf 394}, (2008), 175-186.

\bibitem{fo:5} Fouch\'e, W. L., Fractals generated by algorithmically random Brownian motion, K. Ambos-Spies, B. L\"owe, and W. Merkle (Eds.): CiE 2009, LNCS {\bf 5635} (2009), 208-217.
                    

\bibitem{fomuda} Fouch\'e, W. L., Mukeru, S., and Davie, G., Fourier spectra of measures associated
with algorithmically random Brownian motion, {\it Log. Meth. Comput. Sci}, 10(3:20), (2014), 1-24.

\bibitem{Geman_Horowitz} Geman, D. and Horowitz, J. Occupations densities, 
Ann. Prob. {\bf 8}(1), (1980), 1-67.

\bibitem{horo:1} Hoyrup, M. , Rojas, C., Computability of probability measures and Martin-L\"of randomness over
metric spaces, {\em Information and Computation} {\bf 207}, (2009), 830-847.

\bibitem{horo:2} Hoyrup, M. , Rojas, C., Applications of effective probability theory to Martin-L\"of randomness. In {\it Automata, languages and programming, Part I,} volume 5555 of {\it Lectures Notes in Comput. Sci}.,pp 549-561. Springer, Berlin, 2009. 

\bibitem{ito_mckean} It\^o, K. and McKean, H.P., {\em Diffusion processes and their sample paths.} Springer-Verlag, 1974.

\bibitem{Karandikar_1995} Karandikar, R.L., On pathwise stochastic integration. {\em Stoch. Proc. Appl.} {\bf 57} (1995), 11-18.

\bibitem{KHNe:1} Kjos-Hanssen, B., Nerode, A., The law of the iterated logarithm for algorithmically random Brownian motion, in: Proceedings on Logical
 Foundations of Computer Science, LFCS 2007, in: Lecture Notes in Computer Science, {\bf 4514} (2007), 310-317.
\bibitem{KHNe:2} Kjos-Hanssen, B. and Nerode, A., Effective dimension of points visited by Brownian motion , {\em Theoret. Comput. Sci.} {\bf 410} (2009), 347-354.

\bibitem{KHSZ:01} Kjos-Hansen, B. and Szabados, T., Kolmogorov complexity and strong approximation of Brownian motion, {\em Proc. Amer. Math. Soc.} {\bf 139}  (2011), 3307-3316.
\bibitem{Levy_48} L\'evy P., {\it Processus stochastiques et mouvement Brownien.} Gauthier-Villars, 1948.

\bibitem{Li_Vitanyi} Li, M. and Vitanyi, P., {\it An introduction to Kolmogorov complexity and its applications}, Springer-Verlag, 2008

\bibitem{marlo:1} Martin-L\"{o}f, P., The definition of random sequences, {\em Information and Control} {\bf 9} (1966), 602-619.

\bibitem{McKean_69} McKean, H.P., {\it Stochastic integrals}, Academic Press, 1969. 

\bibitem{Morters_Peres} M\"orters, P. and Peres, Y., Brownian motion, Cambridge University Press, 2010. 


\bibitem{Mukeru_JOC_2014} Mukeru, S. The descriptive complexity of stochastic integration, {\em Journ. Complexity} {\bf 31} (2015), 57-74. 



\bibitem{nie:1} Nies, A. {\it Computability and randomness}, Oxford Logic Guides 51, Clarendon Press, 2008.


\bibitem{Pathak_Rojas_Simpson} Pathak, N., Rojas, C. and Simpson, S.G., Schnorr randomness and the Lebesgue differentiation theorem, {\em Proc. Amer. Math. Soc.}  {\bf 142 (1)} (2014), 335-349.


\bibitem{Rute} Rute, J., {\em Topics in algorithmic randomness and computable analysis}, Thesis, Carnegie Mellon University, (2013).

\bibitem{Trotter} Trotter, H. A property of Brownian motion paths, {Illinois J. Math.} {\bf 2} (1958), 425-433.

\bibitem{Xiao} Xiao, Y. Recent Developments on Fractal Properties of
Gaussian Random Fields, in J. Barral and S. Sueret (eds.), {\it Further Developments in Fractals and Related Fields}, Trends in Mathematics, Springer, 2013. 

\end{thebibliography}
\end{document}